\documentclass[a4paper]{easychair}
\pdfoutput=1 
\usepackage{amsmath}
\usepackage{amsthm}
\usepackage{graphicx}
\usepackage{calc}
\usepackage{hyperref}
\usepackage{wrapfig}
\usepackage{multicol}
\usepackage{color}
\usepackage{import}

\newcommand{\mline}[1]{\begin{array}{c}#1\end{array}}
\newcommand{\outputs}{\mathsf{outputs}}
\newcommand{\FF}{{\cal F}}

\newcommand{\BB}{{\cal B}}
\newcommand\tuple[1]{\langle #1 \rangle}
\newcommand{\sset}[2]{\left\{~#1  \left|
      \begin{array}{l}#2\end{array}
    \right.     \right\}}

\newcommand{\algorithm}[1]{\mbox{\sc\tt{#1}}}

\newcommand{\cinfty}{\multicolumn{1}{c||}{$\infty$}}

\newtheorem{definition}{Definition}
\newtheorem{lemma}{Lemma}

\begin{document}

\title{Twenty-Five Comparators is Optimal\\when Sorting Nine Inputs
  (and Twenty-Nine for Ten)\thanks{Supported by the Israel Science
    Foundation, grant 182/13 and by the Danish Council for Independent
    Research, Natural Sciences. Computational resources provided by an
    IBM Shared University Award (BGU).}}
\author{Michael Codish\inst{1} \and Lu\'{i}s Cruz-Filipe\inst{2} \and Michael Frank\inst{1} \and Peter Schneider-Kamp\inst{2}}
\institute{Department of Computer Science, Ben-Gurion University of the Negev, Israel\\
\texttt{\{\href{mailto:mcodish@cs.bgu.ac.il}{mcodish},\href{mailto:frankm@cs.bgu.ac.il}{frankm}\}@cs.bgu.ac.il} \\[1ex]\and Department of Mathematics and Computer Science, University of Southern Denmark, Denmark\\
\texttt{\{\href{mailto:lcf@imada.sdu.dk}{lcf},\href{mailto:petersk@imada.sdu.dk}{petersk}\}@imada.sdu.dk}}

\maketitle

\begin{abstract}
  This paper describes a computer-assisted non-existence proof of
  9-input sorting networks consisting of 24 comparators, hence
  showing that the 25-comparator sorting network found by Floyd in
  1964 is optimal. As a corollary, we obtain that the 29-comparator
  network found by Waksman in 1969 is optimal when sorting 10 inputs.
  This closes the two smallest open instances of the optimal-size sorting
  network problem, which have been open since the results of Floyd and
  Knuth from 1966 proving optimality for sorting networks of
  up to 8 inputs.
  The proof involves a combination of two methodologies: one based on
  exploiting the abundance of symmetries in sorting networks, and the
  other based on an encoding of the problem to that of satisfiability
  of propositional logic. We illustrate that, while each of these can
  single-handedly solve smaller instances of the problem, it is their
  combination that leads to the more efficient solution that scales to
  handle 9 inputs.
\end{abstract}

\section{Introduction}

General-purpose sorting algorithms are based on comparing, and
possibly exchanging, pairs of inputs. If the order of these
comparisons is predetermined by the number of inputs to sort and does
not depend on their concrete values, then the algorithm is said to be
data-oblivious.  Such algorithms are well suited for e.g.\ parallel
sorting or secure multi-party computations.
  
Sorting networks are a classical formal model for such algorithms
\cite{Knuth73}, where $n$ inputs are fed into networks of $n$
channels, which are connected pairwise by comparators.  Each
comparator takes the two inputs from its two channels, compares them,
and outputs them sorted back to the same two channels. Consecutive
comparators can be viewed as a ``parallel layer'' if no two touch
the same channel.  A comparator network is a sorting network if the
output on the $n$ channels is always the sorted sequence of the
inputs.

Ever since sorting networks were introduced, there has been a quest to
find optimal sorting networks for 
specific given numbers of inputs: optimal size (minimal number of
comparators) as well as optimal depth (minimal number of layers)
networks.  In this paper we focus on optimal-size sorting networks.

Optimal-size and optimal-depth sorting networks for $n \leq 8$ can
already be found in Section 5.3.4 of \cite{Knuth73}.  For optimal
depth, in 1991 Parberry~\cite{DBLP:journals/mst/Parberry91} proved
optimality results for $n = 9$ and $n = 10$, which in 2014 were
extended by Bundala and Z{\'a}vodn{\'y}~\cite{DBLP:conf/lata/BundalaZ14}
to $11 \leq n \leq 16$.
Both approaches are based on breaking symmetries among the first (two)
layers of comparators.

For optimal size, the case of $n=9$ has been the smallest open problem
ever since Floyd and Knuth's result for optimal-size sorting networks
\cite{Knuth66} in 1966. At first, this might be surprising: is optimal size
really harder than optimal depth? However, a
comparison of the sizes of the search spaces for the optimal-size and
optimal-depth problems for $n = 9$ sheds some light on the issues. The
smallest known sorting network for $9$ inputs has size $25$. For
proving/disproving its optimality, we need to consider all comparator
networks of 24 comparators. There are $36 = (9\times 8)/2$
possibilities to place each comparator on $2$ out of $9$ channels.
Thus, the search space for the
optimal-size problem on $9$ inputs consists of $\left(36\right)^{24}
\approx 2.2 \times 10^{37}$ comparator networks.
  
In comparison, to show that the optimal-depth sorting network for $9$
inputs is $7$, one must show that there are no sorting networks of
depth $6$. The number of ways to place comparators in an $n$ channel
layer corresponds to the number of matchings in a complete graph with
$n$ nodes~\cite{DBLP:conf/lata/BundalaZ14}, and for $n=9$ this number
is $2{,}620$.
Thus, the search space for the optimal-depth problem on $9$ inputs is
``just'' $2{,}620^6 \approx 3.2 \times 10^{20}$. In addition, the
layering allows for some beautiful symmetry breaking
\cite{DBLP:conf/lata/BundalaZ14,ourarxivstuff} on the first two
layers, reducing the search space further to approx.\ $10^{15}$
comparator networks.

For the optimal-depth problem, all recent attempts
we are aware of~\cite{DBLP:conf/mbmv/MorgensternS11,DBLP:conf/lata/BundalaZ14} have
used encodings to the satisfiability problem of propositional logic
(SAT). Likewise, in this paper we describe a SAT encoding for the
optimal-size problem. This SAT encoding was able to reproduce all
known results for $n \leq 6$.  Unfortunately, the SAT encoding alone
did not scale to $n = 9$, with state-of-the-art SAT solvers making no
discernible progress even after weeks of operation.
  
To solve the open problem of optimality for $n = 9$, we had to invent
symmetry breaking techniques for reducing the search space to a
manageable size. The general idea is similar to the one taken in
\cite{DBLP:conf/lata/BundalaZ14,ourarxivstuff} 
for the optimal-depth sorting network problem, but involves the
generation of minimal sets of non-redundant comparator networks for a
given number of comparators, one comparator at a time. Redundant
networks (i.e., networks that sort less than others of same size or
that are equivalent to another network already in the set) are
pruned. For each pruned network, a 
witness is logged, which can be independently verified.
  
For $n = 9$, we used this method, which we call generate-and-prune, to
reduce the search space from approx.\ $2.2 \times 10^{37}$
to approx.\ $3.3 \times 10^{21}$ comparator
networks, all of which can be obtained by extending one of $914{,}444$
representative 14-comparator networks.  This process took a little
over one week 
of computation, and all of the resulting problems could be handled
efficiently by our SAT encoding in less than $12$ hours (in
total). All computations, if not specified otherwise, were performed
on a cluster with a total of 144 Intel E8400 cores clocked at 3 GHz
each and able to run a total of $288$ threads in parallel.
  
The generate-and-prune method can also be used in isolation to decide
this open problem: amongst the set of all comparator networks (modulo
equivalence and non-redundancy) there is only one single sorting
network, and it is of size 25.  To obtain this result, we continued
running the generate-and-prune method for five more days in order to
check the validity of the results obtained through the SAT encoding
independently, thereby instilling a higher level of trust into the
computer-assisted proof.  This paper presents 
both techniques: the first one based completely on the generate-and-prune
approach, and the second, hybrid, method combining generate-and-prune
with SAT encoding. It is the second approach that solves the
nine-input case in the least amount of time, and also shows the potential to
scale.
  
Once determining that 25 comparators 
is optimal for 9 inputs, we move on to consider the case of $10$
inputs.  Using a result of van Voorhis from 1971~\cite{voorhis72}, we
know that the minimal number of comparators for sorting 10 inputs is
at least $4$ larger than for 9 inputs.  As a sorting network with
$29$ comparators on ten inputs (attributed to Waksman) is known since
1969~\cite{Knuth73}, our result implies its optimality.
  
The next section introduces the relevant concepts on sorting networks
together with some notation.  The generate-and-prune algorithm is
introduced in Section~\ref{sec:step1}, while its optimization and
parallelization are discussed in detail in
Section~\ref{sec:implement}.  The SAT encoding is explained and
analyzed in Section~\ref{sec:step2}.  In Section~\ref{sec:discussion}
we reflect on the validity of the proof, and we conclude in
Section~\ref{sec:concl}.
 
\section{Preliminaries on sorting networks}
\label{sec:prelim}

A \emph{comparator network} $C$ with $n$ channels and size $k$ is a
sequence of \emph{comparators} $C = (i_1,j_1);\ldots;(i_k,j_k)$ where
each comparator $(i_\ell,j_\ell)$ is a pair of channels $1\leq i_\ell
< j_\ell\leq n$. The \emph{size} of a comparator network is the number
of its comparators. If $C_1$ and $C_2$ are comparator networks, then
$C_1;C_2$ denotes the comparator network obtained by concatenating
$C_1$ and $C_2$; if $C_1$ has $m$ comparators, it is a \emph{size-$m$
  prefix} of $C_1;C_2$.
An input $\vec x=x_1\ldots x_n\in\{0,1\}^n$ propagates through $C$ as
follows: $\vec x^0 = \vec x$, and for $0<\ell\leq k$, $\vec x^\ell$ is
the permutation of $\vec x^{\ell-1}$ obtained by interchanging $\vec
x^{\ell-1}_{i_\ell}$ and $\vec x^{\ell-1}_{j_\ell}$ whenever $\vec
x^{\ell-1}_{i_\ell}>\vec x^{\ell-1}_{j_\ell}$.  The output of the
network for input $\vec x$ is $C(\vec x)=\vec x^k$, and
$\outputs(C)=\sset{C(\vec x)}{\vec x\in\{0,1\}^n}$.  The comparator
network $C$ is a \emph{sorting network} if all elements of
$\outputs(C)$ are sorted (in ascending order).
\begin{wrapfigure}[7]{r}{0.25\textwidth}
  \vspace*{-0.5ex}
  \quad$(a)$~\raisebox{-\height/2}{\includegraphics{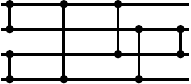}}

  \vspace{1em}

  \quad$(b)$~\raisebox{-\height/2}{\includegraphics{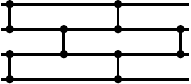}}
\end{wrapfigure}
The zero-one principle~(e.g.~\cite{Knuth73}) implies that a sorting
network also sorts 
sequences over any other totally ordered set, e.g.~integers.
Images~$(a)$ and~$(b)$ on the right depict sorting networks on 4
channels, each consisting of 6 comparators.  The channels are
indicated as horizontal lines (with channel $4$ at the bottom),
comparators are indicated as vertical lines connecting a pair of
channels, and input values are assumed to propagate from left to
right. The sequence of comparators associated with a picture
representation is obtained by a left-to-right, top-down
traversal. For example the networks depicted above are: $(a)$
$(1,2); (3,4); (1,4); (1,3); (2,4); (2,3)$ 
and $(b)$~$(1,2); (3,4); (2,3); (1,2);(3,4); (2,3)$.

The optimal-size sorting network problem is about finding the smallest
size, $S(n)$, of a sorting network on $n$ channels.  In~\cite{Knuth66},
Floyd and Knuth present sorting networks of optimal size for $n\leq 8$
and prove their optimality. Until today, the minimal size $S(n)$ of a
sorting network on $n$~channels
was
known only for $n\leq 8$; for
greater values of $n$, there are upper bounds on $S(n)$ obtained e.g.\
by the systematic construction of
Batcher~\cite{DBLP:conf/afips/Batcher68}, or by concrete examples of
sorting networks (see~\cite{Knuth73}).  The previously best known upper and lower bounds
for $S(n)$ are given in~\cite{Knuth66} and reproduced in the first two lines of the table below.
The last line shows the contribution of this paper, i.e., the improved lower bounds, matching the upper bounds for $n = 9$ and $n = 10$.

{\small
\[
\begin{array}{l|c|c|c|c|c|c|c|c|c|c|c|c|c|c|c|c}
\!\!\!n   & 1 & 2 & 3 & 4 & 5 & 6  & 7  & 8  & 9  & 10 & 11 & 12 & 13 & 14 & 15 & 16\!\!\! 
\\ \hline
\!\!\!\mbox{upper bound} & 0 & 1 & 3 & 5 & 9 & 12 & 16 & 19 & 25 & 29 & 35 & 39 & 45 & 51 & 56 & 60\!\!\!\\
\!\!\!\mbox{old lower bound}\,&\,0\,&\,1\,&\,3\,&\,5\,&\,9\,&\,12\,&\,16\,&\,19\,&\,23\,&\,27\,&\,31\,&\,35\,&\,39\,&\,43\,&\,47\,&\,51\!\!\!\\
\!\!\!\mbox{new lower bound}\!&\,\,&\,\,&\,\,&\,\,&\,\,&\,\,&\,\,&\,\,&\,25\,&\,29\,&\,33\,&\,37\,&\,41\,&\,45\,&\,49\,&\,53\!\!\!
\end{array}
\]}

The following lemma due to van Voorhis \cite{voorhis72} can be used to
establish lower bounds for $S(n)$.
\begin{lemma}
\label{lem:vanvoorhis}
$S(n+1) \geq S(n) + \lceil\log_2n\rceil$ for every $n \geq 1$.
\end{lemma}
This lemma was applied in \cite{Knuth66} to derive the values of
$S(6)$ and $S(8)$ from those of $S(5)$ and $S(7)$,
respectively. Likewise, we apply Lemma~\ref{lem:vanvoorhis} to obtain
the value of $S(10)$ from our proof that $S(9) = 25$ and,
consequently, we are able to improve the values for $S(n)$ for $n>10$,
as indicated in the third line of the 
above table.

Crucial to our approach is the exploitation of symmetries in
comparator networks, and these can be expressed in terms of
permutations on channels.
Given an $n$ channel comparator network
$C=(i_1,j_1);\ldots;(i_k,j_k)$, and a permutation $\pi$ on
$\{1,\ldots,n\}$, $\pi(C)$ is the sequence
$(\pi(i_1),\pi(j_1));\ldots;(\pi(i_k),\pi(j_k))$.  Formally, $\pi(C)$
is not a comparator network, but rather a generalized comparator
network.
A \emph{generalized comparator network} is defined like a comparator
network, except that it may contain comparators $(i,j)$ with $i>j$,
which order their outputs in descending order, instead of ascending.
It is well-known~(e.g.~Exercise~5.3.4.16 in \cite{Knuth73})
that generalized sorting networks are no more powerful than sorting
networks: a generalized sorting network can always be \emph{untangled}
into a (standard) sorting network with the same size and depth.

We write $C_1 \approx C_2$ ($C_1$ is equivalent to $C_2$) iff there is a
permutation $\pi$ such that $C_1$ is obtained by untangling the (generalized)
comparator network $\pi(C_2)$.
The two networks~$(a)$ and~$(b)$ above are equivalent via the permutation
$(1\,3)(2\,4)$ and the application of the construction for untangling
described in~\cite{Knuth73} (Exercise~5.3.4.16).

Another important and related concept is that of a complete set of
filters for the optimal-size sorting network problem.

\begin{definition}[complete set of filters]
\label{def:complete}
We say that a (finite) set, $\FF$, of comparator networks on $n$ channels
is a \emph{complete set of filters} for the optimal-size sorting
network problem on $n$ channels if it is the case that there exists an
optimal-size sorting network on $n$ channels if and only if there
exists one of the form $C;C'$ for some $C\in\FF$.
\end{definition}

For any given $n$ there always exists a complete set of filters:
simply take the set of all comparator networks on $n$ channels. In
this paper we will focus on the search for ``small'' complete sets in
which all filters are of the same size.

\section{The generate-and-prune approach}
\label{sec:step1}

In this section we consider the task of generating the set of all
$n$-channel comparator networks consisting of $k$ comparators. Given
this set one could, at least conceptually, inspect the networks
one-by-one to determine if there exists an $n$-channel, $k$-comparator,
sorting network. Obviously, such a naive approach is
combinatorically infeasible. With $n$ channels,
there are $n(n-1)/2$ possibilities for each comparator,
and thus incrementally adding comparators would produce
$(n(n-1)/2)^k$ networks with $k$ comparators. For $n=9$, aiming
to prove that there does not exist a sorting network with 24
comparators would mean inspecting approximately $2.25\times10^{37}$
comparator networks.
Moreover, checking whether a comparator network is a sorting network
is known to be a co-NP complete problem~\cite{DBLP:conf/parle/Parberry91}.

We propose an alternative approach, \emph{generate-and-prune}, which
is driven just as the naive approach, but takes advantage of the
abundance of symmetries in comparator networks. It is best described
after introducing a definition and a lemma.

\begin{definition}[subsumption]
Let $C_a$ and $C_b$ be comparator networks on $n$ channels. If there
exists a permutation $\pi$ such that $\pi(\outputs(C_a)) \subseteq
\outputs(C_b)$ then we denote this as $C_a\leq_{\pi}C_b$ and we say
that $C_a$ \emph{subsumes} $C_b$. We also write $C_a\preceq C_b$ to
indicate that there exists a permutation $\pi$ such that
$C_a\leq_{\pi}C_b$. 
\end{definition}
Observe that $\preceq$ is a reflexive and transitive relation, and
that ${\approx}\subseteq{\preceq}$.

\begin{lemma}
  \label{lem:outputs}
  Let $C_a$ and $C_b$ be comparator networks on $n$ channels, both of
  the same size, and such that $C_a\preceq C_b$. Then, if there exists
  a sorting network $C_b;C$ of size~$k$, there also exists a sorting
  network $C_a;C'$ of size~$k$.
\end{lemma}
\begin{proof}
  Under the hypotheses, there exists a permutation $\pi$ such that
  $C_a\leq_{\pi}C_b$.  Untangling $C_b;\pi^{-1}(C)$ into $C_b;C'$
  yields the desired sorting network (see the proof of the similar Lemma~7
  in~\cite{DBLP:conf/lata/BundalaZ14} for details).
\end{proof}

Lemma~\ref{lem:outputs} implies that, when adding a next comparator
in the naive approach, we do not need to consider all possible
positions to place it. In particular, we can omit networks which are
subsumed by others.

The \emph{generate-and-prune} algorithm is as follows, where $R^n_k$
and $N^n_k$ are sets of $n$ channel networks each consisting of $k$
comparators. First, initialize the set $R^n_0$ to consist of a single
element: the empty comparator network. Then, repeatedly apply two
types of steps, \algorithm{Generate} and \algorithm{Prune}, to add
comparators in all possible ways incrementally,
and then remove those subsumed by others.

\begin{enumerate}
\item\algorithm{Generate}: Given the set $R^n_k$, derive the set $N^n_{k+1}$
  containing all nets obtained by adding one extra comparator to each
  element of $R^n_k$ in all possible ways.
\item\algorithm{Prune}: Given the set $N^n_{k+1}$, derive the set
  $R^n_{k+1}$ obtained by pruning $N^n_{k+1}$ to remove networks
  subsumed by those which are not pruned.
\end{enumerate}

The pruning step can thus be described as keeping only one network producing
each minimal set of outputs (under permutation).
In other words, it keeps one representative of each equivalence class of
minimal networks w.r.t.\ $\preceq$, independently of the order in which
the subsumption tests are made.

\begin{lemma}\label{lem:complete}
  For every $n$ and $k$, the sets $N^n_k$ and $R^n_k$ are complete
  sets of filters on $n$ channels.
\end{lemma}

Note that if a set of networks includes a sorting
network, then pruning that set will leave precisely one element (a
sorting network).

The \algorithm{Generate} and \algorithm{Prune} algorithms, shown in
Figure~\ref{fig:algorithms}, are both very simple.  However, they
operate on huge data sets, consisting of millions of comparator
networks.  So, it is the small implementation details that render them
computationally feasible.  We first describe their schematic
implementation and then describe some of their finer details.

\begin{figure}
\begin{multicols}{2}
\paragraph{Algorithm \algorithm{Generate}.}\ \\
\begin{tabbing}
\textsf{input:} $R^n_k$;  \quad \textsf{output:} $N^n_{k+1}$; \\
$N^n_{k+1}=\emptyset$; \\
$C_n = \sset{(i,j)}{1\leq i<j\leq n}$ \\
\textsf{for} \= $C\in R^n_k$ \textsf{and} $c\in C_n$ \textsf{do} \+ \\
  $N^n_{k+1}=N^n_{k+1}\cup\{C;c \}$; \\ \\ \\ \\
\end{tabbing}
\paragraph{Algorithm \algorithm{Prune}.}

\begin{tabbing}
\textsf{input:} $N^n_k$; \quad \textsf{output:} $R^n_k$; \\
$R^n_k=\emptyset$; \\
\textsf{for} \= $C\in N^n_k$ \textsf{do} \+ \\
  \textsf{for} \= $C'\in R^n_k$ \textsf{do} \+ \\
    \textsf{if} ($C'\preceq C$)  mark $C$; \- \\
  \textsf{if} (not\_marked($C$))  \+ \\
    $R^n_k =R^n_k \cup \{C\}$;  \\
    \textsf{for} \= $C'\in R^n_k$ \textsf{do} \+ \\
      \textsf{if} ($C\preceq C'$)  $R^n_k =R^n_k \setminus \{C'\}$;
\end{tabbing}
\end{multicols}
\caption{The \algorithm{Generate} and \algorithm{Prune} algorithms.}
\label{fig:algorithms}
\end{figure}

The \algorithm{Generate} algorithm takes a set, $R^n_k$, of networks,
and adds to each network in the set one new comparator in
every possible way. There are $n(n-1)/2$ ways to add a comparator
on $n$ channels, hence, the execution time of \algorithm{Generate}
is $O\left(n^2\times\left|R^n_k\right|\right)$.

The \algorithm{Prune} algorithm basically tests each network from its input, $N^n_k$, keeping only those networks which are not subsumed by any other network encountered so far. These minimal (w.r.t.\ subsumption) networks are kept in the set $R^n_k$, which after execution of the algorithm contains a complete set of filters on $n$ channels.
The sets $R^n_k$ are initially empty, and then they grow and shrink throughout the run of the algorithm, until finally containing only minimal elements in the order $\preceq$.
While theoretically $R^n_k$ could first grow to nearly $|N^n_k|$ before
collapsing to its final size, experimentation indicates that the
intermediate sizes of $R^n_k$ are bounded by its final size.
Thus, the algorithm is posed such that the outer loop is on the elements of
$N^n_k$, and the inner loop on the current set $R^n_k$.

In this manner, the worst-case behavior of \algorithm{Prune}
is $O\left(\left|N^n_k\right|\times \left|R^n_k\right|
  \times f(n)\right)$, where $f(n)$ is the cost of a single subsumption
test.
A naive implementation tests if $C_a\preceq C_b$ maintaining the
sets $S_a= \outputs(C_a)$ and $S_b= \outputs(C_b)$ and iterating over
the space of $n!$ permutations to test if there exists a permutation $\pi$
such that $\pi(S_a) \subseteq S_b$.

These very simple algorithms are straightforward to implement, test
and debug. Our implementation is written in Prolog and can be applied to
reconstruct all of the known values for $S_n$ for $n\leq 6$ in under
an hour of computation on a single core.
The table below shows the values for $\left|R^n_k\right|$ when $n\leq
8$; the values for $n=7,8$ were obtained using the optimized version of
our implementation described in the next sections.  For any $k$,
if there is no sorting network on $n$~channels with $k$~comparators,
then $\left|R^n_k\right|>1$, since a sorting network trivially subsumes any
other comparator network.  Recall also that
$\left|N^n_k\right|=\frac{n(n-1)}2\left|R^n_{k-1}\right|$.

{\small%
\[\begin{array}{r|r|r|r|r|r|r|r|r|r|r|r|r|r|r|r|r|r|r|r}
\!\!n\backslash k
& \multicolumn{1}{c|}{\!1\!}
& \multicolumn{1}{c|}{\!2\!}
& \multicolumn{1}{c|}{\!3\!}
& \multicolumn{1}{c|}{4}
& \multicolumn{1}{c|}{5}
& \multicolumn{1}{c|}{6}
& \multicolumn{1}{c|}{7}
& \multicolumn{1}{c|}{8}
& \multicolumn{1}{c|}{9}
& \multicolumn{1}{c|}{10}
& \multicolumn{1}{c|}{11}
& \multicolumn{1}{c|}{12}
& \multicolumn{1}{c|}{13}
& \multicolumn{1}{c|}{14}
& \multicolumn{1}{c|}{15}
& \multicolumn{1}{c|}{\!16\!}
& \multicolumn{1}{c|}{\!17\!}
& \multicolumn{1}{c|}{\!18\!}
& \multicolumn{1}{c}{\!19\!} \\ \hline
2&\!1\!&&&&&&&&&&&&&&&&& \\
3&\!1\!&\!2\!&\!1\!&&&&&&&&&&&&&&& \\
4&\!1\!&\!3\!&\!4\!&2&1&&&&&&&&&&&&& \\
5&\!1\!&\!3\!&\!6\!&\!11\!&\!10\!&7&6&4&1&&&&&&&&&& \\
6&\!1\!&\!3\!&\!7\!&\!17\!&\!36\!&53&53&44&23&8&4&1&&&&&&&\\
7&\!1\!&\!3\!&\!7\!&\!19\!&\!51\!&\!141\!&\!325\!&564&678&510&280&106&33&11&6&1&&&\\
8&\!1\!&\!3\!&\!7\!&\!20\!&\!57\!&\!189\!&\!648\!&\!2088\!&\!5703\!&\!11669\!&\!16095\!&\!13305\!&\!6675\!&\!2216\!&\!503\!&\!77\!&\!18\!&9&1\\
\end{array}\]}

We analyze the case $n=7$ in some detail.  There are $21$
possibilities for the first comparator $(i,j)$ on a $7$-channel comparator
network; however, these are all equivalent by means of the permutation
$(i\;1)(j\;2)$.  Hence $\left|R^7_1\right|=1$.  We assume the single
representative to be the network $(1,2)$.
The second comparator can again be one of the same $21$; but there are
only four possibilities that are not equivalent: either it is again
$(1,2)$, or it is of the form $(1,j)$ with $j\neq 2$, or of the form
$(2,j)$ with $j>2$, or of the form $(i,j)$ with $2<i<j$.  The first
possibility yields a comparator network that is subsumed by any of the
others. 
For the other three possibilities, suitable permutations can map the
second comparator to $(1,3)$, $(2,3)$ or $(3,4)$, respectively.
Therefore, $\left|R^7_2\right|=3$, and the representatives can be
chosen to be net $(1,2);(1,3)$, net $(1,2);(2,3)$ and net
$(1,2);(3,4)$.
A similar reasoning shows that there are only seven possibilities for the
three-comparator networks, and a representative set contains e.g.:

\begin{multicols}{4}
\begin{itemize}
\item $(1,2);(2,3);(1,2)$
\item $(1,2);(3,4);(1,3)$
\item $(1,2);(3,4);(1,4)$
\item $(1,2);(3,4);(1,5)$
\item $(1,2);(3,4);(2,4)$
\item $(1,2);(3,4);(2,5)$
\item $(1,2);(3,4);(5,6)$
\end{itemize}
\end{multicols}

\section{Implementing generate-and-prune}
\label{sec:implement}
This section describes details of the implementation of the
\algorithm{Generate} and \algorithm{Prune} algorithms and the
optimizations that, in the end, make it possible to compute the
precise value of $S(9)=25$.
Here we keep in mind that the values for $n^2$, $2^n$, and $n!$ where
$n=9$ are constants: $81$, $512$, and $362{,}880$. On the other hand,
the number of elements in $\left|N^9_{24}\right|$
could potentially grow to more than $10^{37}$.

\subsection{Representing comparator networks}

The inner loops in the \algorithm{Prune} algorithm involve subsumption
tests on pairs of networks. We implement these in terms of the search
for a permutation under which the outputs of the one network are a subset
of the outputs of the other. Moreover, as each network is
tested for subsumption multiple times,
we choose to represent a comparator network, explicitly, together with
the set of its outputs.  It is convenient to represent the output binary
sequence $\vec x=x_1\ldots x_n$ by the corresponding binary number
(least significant bit first), $\#\vec x$.
With this representation, $x_i=(\#\vec x/2^{i-1}\bmod 2)$, where
`$/$'
stands for integer division, and the result of exchanging positions
$i$ and $j$ in $\vec x$ translates to computing $\#\vec
x-2^{i-1}+2^{j-1}$ when $x_i=1$ and $x_j=0$, the only case when such
an exchange is necessary.  These operations can be implemented extremely efficiently, e.g.\ using shifts.

As an example, consider the comparator network $C = (1,2);(3,4);(1,3)$
on four channels with $\outputs(C) = \{0000, 0001, 0011, 0100, 0110,
0101, 0111,1111\}$, represented as the set $\{0, 8, 12, 2, 6, 10, 14,
15\}$.  Consider the output $\vec x=0101$, for which $\#\vec x=10$.  We
have $x_1=(10/2^0\bmod 2) = 0$ and $x_2=(10/2^1\bmod 2) = %
1$, and likewise $x_3=0$ and $x_4=1$.  Since $x_2>x_3$, applying the
comparator $(2,3)$ to $\vec x$ yields the sequence $\vec y$ such that
$\#\vec y=\#\vec x-2^1+2^2= 12$, namely the sequence $0011$.
In the same way, it is easy to check that $\outputs(C;(2,3))$ is
represented as the set $\{0,8,12,4,6,14,15\}$.

Given this choice, in \algorithm{Generate}, adding a
comparator $(i,j)$ to a network $C$ simply requires applying
$(i,j)$ to those elements $\#\vec x$ of the set of
outputs in the representation of $C$ for which $x_i>x_j$.
So, the cost of computing output sets
\emph{diminishes} with each extra comparator, since the sizes of the
output sets decrease with each addition.
In the example above, adding the comparator $(2,3)$ to the network
would change $10$ to $12$ and $2$ to $4$.

The \algorithm{Generate} algorithm is implemented to produce a file
where each network is tupled with the set of its outputs (represented
as numbers) and some additional information that is detailed
below. 
Moreover, the elements in these sets are partitioned according to the
number of ones their binary representation contains,
as this facilitates the optimizations
described below. For instance, in the context of the previous example,
we represent $C$ as the following triplet, where $W$ is described in
the next section.
\begin{equation}
\label{eq:triplet}
\left\langle \{(1,2);(3,4);(1,3)\},
  \left\{ \{0\}, \{2, 8\}, \{6, 10, 12\},
    \{14\},\{15\}\right\}, W \right\rangle
\end{equation}

Even though we are adding extra information exponential in $n$, this
is still manageable in practice.  Case in point, the largest file
encountered in the proof of $n=9$ contains $N^9_{15}$ and is just
under $11$ GB in size. We need to keep at most two files at any given
point of time: to support pruning of $N^9_k$ to $R^9_k$, and to
support extending $R^9_k$ to $N^9_{k+1}$.

\subsection{Implementing the test for subsumption}
\label{subsection:subsumption}

We implemented the subsumption test $C_1\preceq C_2$ in
\algorithm{Prune} as the search problem of finding a permutation
$\pi$ such that $\pi(\outputs(C_1))\subseteq \outputs(C_2)$. For $9$
channels, this might involve considering $362{,}880$ permutations.
We illustrate why, in many cases, it is computationally easy to detect
the non-existence of such a permutation, and how we restrict the
search space considerably in the other cases.
This optimization is crucial to move beyond the case of $6$~channels.

Let $S_1=P_0 \uplus \ldots \uplus P_n$ and $S_2= Q_0 \uplus \ldots \uplus Q_n$ be two sets
of length-$n$ binary sequences partitioned according to the number of
$1$s in the sequences. A basic observation that can be applied to
refine the search for a suitable permutation $\pi$ is that
$\pi(S_1)\subseteq S_2$ if and only if $(\pi(P_0)\subseteq
Q_0)\land\cdots\land(\pi(P_n)\subseteq Q_n)$.
Moreover, there are several easy-to-check criteria which apply to
determine that no such $\pi$ exists.  We introduce these through an
example.  

Figure~\ref{example:perm} details three $5$-channel comparator
networks together with their sets of outputs partitioned according to
their number of ones. Focusing on the column detailing the output
sequences with two 1s, it is clear that $C_2\not\preceq C_1$. Indeed,
any permutation of $\outputs(C_2)$, must have four sequences with two $1$s
each, and hence $\pi(\outputs(C_2))$ cannot be a subset of
$\outputs(C_1)$, which has only three sequences with two $1$s.
The same type of argument implies that $C_2\not\preceq C_3$,
$C_3\not\preceq C_1$ (looking at outputs with four $1$s) and
$C_3\not\preceq C_2$ (looking at outputs with one $1$).

\begin{figure}[t]
\centering
\[\begin{array}{c|c|c|c|c|c|c}
\mbox{no.\;of $1$s:} & 0 & 1 & 2 & 3 & 4 & 5 \\ \hline
C_1~\raisebox{-\height/2}{\includegraphics{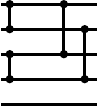}}
& 00000
& \mline{00001\\ 00010}
& \mline{00011\\ 00110\\ 01010}
& \mline{00111\\ 01011\\ 01110}
& \mline{01111\\ 11110}
& 11111 \\ \hline
C_2~\raisebox{-\height/2}{\includegraphics{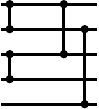}}
& 00000
& \mline{00001\\ 00010}
& \mline{00011\\ 00101\\ 00110\\ 01001}
& \mline{00111\\ 01011\\ 01101}
& \mline{01111\\ 10111}
& 11111 \\ \hline
C_3~\raisebox{-\height/2}{\includegraphics{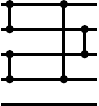}}
& 00000
& \mline{00001\\ 00010\\ 00100}
& \mline{00011\\ 00101\\ 00110}
& \mline{00111\\ 01110\\ 10110}
& \mline{01111\\ 10111\\ 11110}
& 11111
\end{array}\]
  \caption{Three 5-channel comparator networks with their partitioned output sets.}
\label{example:perm}
\end{figure}

More formally, we state the following lemma.
\begin{lemma}
  \label{lem:presolve}
  Let $C_a$ and $C_b$ be $n$ channel comparator networks.  If there
  exists $1\leq k\leq n$ such that the number of sequences with
  $k$ $1$s in $\outputs(C_a)$ is greater than that in $\outputs(C_b)$,
  then $C_a\not\preceq C_b$.
\end{lemma}

Experiments show that, in the context of this paper, more than $70\%$
of the subsumption tests in the application of the \algorithm{Prune}
algorithm are eliminated based on Lemma~\ref{lem:presolve}.

Focusing again on Figure~\ref{example:perm}, this time on the column
detailing the output sequences with three 1s, it becomes clear that
$C_1\not\preceq C_3$.  This is because the digit $0$ occurs in four
different positions in the sequences for $C_1$, and this will remain
the case when applying any permutation to its elements, but only in
three different positions in the sequences for $C_3$.  To formalize
this observation we introduce some notation.  If $C$ is an $n$-channel
comparator network, $x\in\{0,1\}$, and $0\leq k\leq n$ is an integer
value, then $w(C,x,k)$ denotes the 
set of positions $i$ such that there exists a vector $x_1\ldots x_n$
in $\outputs(C)$ containing $k$ ones, and such that $x_i=x$.

\begin{lemma}
  \label{lem:filterPosition}
  Let $C_a$ and $C_b$ be $n$ channel comparator networks. If for some
  $x\in\{0,1\}$ and $0\leq k\leq n$, $|w(C_a,x,k)|>|w(C_b,x,k)|$ then
  $C_a\not\preceq C_b$.
\end{lemma}

Experiments show that, in the context of this paper, around $15\%$ of
the subsumption tests in the application of the \algorithm{Prune}
algorithm that are not eliminated based on  Lemma~\ref{lem:presolve}
are subsequently eliminated by application of Lemma~\ref{lem:filterPosition}.

In order to apply this criterion %
efficiently, the sets $w(C,x,k)$, for $x\in\{0,1\}$ and $0\leq k\leq
n$, are computed when $C$ is generated and maintained as part of the
representation of $C$. This is the third element, $W$, in the triplet of
Equation~\eqref{eq:triplet}.

In the following lemma, we observe that the information in the sets
$w(C,x,k)$ is also helpful in restricting the search space for a
suitable permutation.

\begin{lemma}
  \label{lem:legalPermutation}
  Let $C_a$ and $C_b$ be $n$-channel comparator networks and $\pi$ be a permutation.
  If $\pi(\outputs(C_a))\subseteq\outputs(C_b)$,
  then $\pi(w(C_a,x,k))\subseteq w(C_b,x,k)$ for all $x\in\{0,1\}$, $1\leq k\leq n$.
\end{lemma}

Implementing these optimizations in \algorithm{Prune}
reduces the computation time for $6$~channels by a factor of over
$200$, and allows the verification of the known results for $n=7$ in a
few minutes and for $n=8$ in a few hours.  For $n=7$, the largest set
of reduced networks that has to be considered is $R^7_9$, which
contains $678$ elements.  Of the $33$~million subsumption tests
performed in the whole run, more than $27$~million
were solved by
application of Lemma~\ref{lem:presolve} and another approx.\ $600$ thousand by Lemma~\ref{lem:filterPosition}.

\subsection{Avoiding redundant comparators}

Let us come back to the operation of incrementally adding comparators
as specified in \algorithm{Generate}. In some cases, it
is easy to identify that a comparator is \emph{redundant} and not to
add it in the first place. Networks obtained by adding a redundant
comparator would anyway be removed by \algorithm{Prune},
but that involves the more expensive subsumption test.

Consider a comparator network of the form $C;(i,j);C'$.  We say that
$(i,j)$ is redundant if $x_i \leq x_j$ for all sequences $x_1 \ldots
x_n \in \outputs(C)$.
This notion of redundant comparators is simpler than the one
proposed in Exercise 5.3.4.51 of~\cite{Knuth73} (credited to
R.L.~Graham), but equivalent for standard sorting networks.
Since comparator networks are represented explicitly together with
their output sets, this condition is straightforward to check.

In the loop of \algorithm{Generate}, we refrain from adding redundant
comparators to the networks being extended, thus guaranteeing that
there are no redundant comparators in $R^n_k$.
Correctness of not adding redundant comparators follows in the same way as
in the context of Exercise 5.3.4.51 of~\cite{Knuth73}. Let
$C;(i,j);C'$ be a sorting network obtained by extending $C;(i,j)$.  If
$(i,j)$ is redundant, then $C;C'$ is also a sorting network, and
smaller. Implementing this optimization, 
\begin{wrapfigure}[9]{r}{0.3\textwidth}
  \vspace*{-4ex}
\paragraph{Algorithm \algorithm{Generate$'$}.}\ 
\begin{tabbing}
\textsf{input:} $R^n_k$;  \quad \textsf{output:} $N^n_{k+1}$; \\
$N^n_{k+1}=\emptyset$; \\
$C_n = \sset{(i,j)}{1\leq i<j\leq n}$ \\
\textsf{for} \= $C\in R^n_k$ \textsf{and} $c\in C_n$ \textsf{do} \+ \\
  \textsf{if} ($\neg$\textsf{redundant}($C$,$c$))  \\
  \qquad $N^n_{k+1}=N^n_{k+1}\cup\{C;c \}$;
\end{tabbing}
\end{wrapfigure}
depicted as Algorithm \algorithm{Generate$'$} on the right, the values of
$\left|N^n_k\right|$ drop significantly, especially as $k$ increases.
Typically, the highest value of $\left|N^n_k\right|$ is reduced by
more than $40\%$; subsequent values drop even more, although their
impact on computation time is less pronounced.  As a result, the total
execution time for generate-and-prune is reduced to about one half for
each value of $n\leq 8$.
The size of the largest $\left|N^n_k\right|$ is given in the table
below, for $n=6$, $7$ and $8$, without any optimizations and when
refraining from adding redundant comparators.
\[\begin{array}{c|r|r|r|c}
n & \multicolumn{1}{c|}k & \multicolumn{1}{c|}{\mbox{ original }} & \multicolumn{1}{c|}{\mbox{ no redundancies }} & \mbox{ relative reduction } \\ \hline
6 & 7 & 795 & 457 & 42.5\% \\
7 & 10 & 14{,}238 & 7{,}438 & 47.8\% \\
8 & 12 & 450{,}660 & 253{,}243 & 43.8\%
\end{array}\]

The optimal sorting networks for sizes $5$ to $8$ found by our optimized
generate-and-prune algorithm are given below.

\vspace{1ex}\noindent\hfill
\raisebox{-\height/2}{\includegraphics{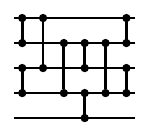}}
\hfill
\raisebox{-\height/2}{\includegraphics{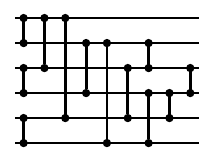}}
\hfill
\raisebox{-\height/2}{\includegraphics{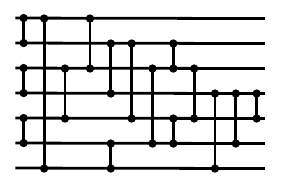}}
\hfill
\raisebox{-\height/2}{\includegraphics{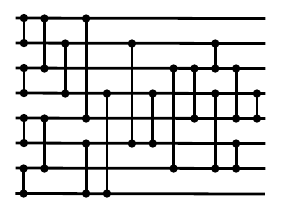}}
\hspace*\fill\vspace{2ex}

The execution of generate-and-prune for $n=9$, $k \in \{0,\ldots,25\}$
remains a daunting task. To see this, consider that the growth of the
values of $\left|N^9_k\right|$ and $\left|R^9_k\right|$ (which for $k
= 14$ turned out to be $18{,}420{,}674$ and $914{,}444$, respectively)
requires more than $10$ trillion subsumption checks, each in the
worst-case requiring to check $9!=362{,}380$ permutations. On a
positive note, the optimizations described up to here allow the
algorithms to be run for $n=9$ within the life span of a human being
(more precisely, an expected approx.\ $9$ years of computation on a
single core).

\subsection{Parallelization}

In order to reduce the \emph{total} execution time of
generate-and-prune for $n=9$, we developed parallelized versions of
both algorithms.  
We consider a distributed-memory architecture consisting of $p$
processing elements.  For \algorithm{Generate}, the parallelization is
straightforward, as the the extension of each network in $R^n_k$ can
be done independently, i.e., the set can be split into $|R^n_k|$
singleton sets, which can be processed by \algorithm{Generate} in
parallel.  In addition, the resulting extensions are all pairwise
different, so set union can be implemented as a simple
merge-by-concatenation of the extensions. As the number of networks to
extend is typically considerably larger than the number of processing
elements $p$, and both splitting and merging incur some overhead, in
practice we divided $R^n_k$ into $p$ sets of equal size. As the
sequential algorithm is linear in $\left|R^n_k\right|$ and there is no
communication overhead in the parallel version, the latter has
constant isoefficiency~\cite{DBLP:journals/ieeecc/GramaGK93}.
Figure~\ref{fig:par-generate} presents a straightforward
parallelization of \algorithm{Generate}, where \textsf{for$||^p$}
indicates a parallel for-each loop using
$p$ processing elements at the same time.

\begin{figure}
\begin{multicols}{2}
\paragraph{Algorithm \algorithm{Parallel-Generate}.}\ \\
\begin{tabbing}
\textsf{input:} $R^n_k$;  \quad \textsf{output:} $N^n_{k+1}$; \\
split $R^n_k$ into sets $R_1,\ldots,R_p$ \\
\textsf{for$||^p$} \= $i \in \{1,\ldots,p\}$ \textsf{do} \+ \\
$S_i = \algorithm{Generate$'$}(R_i)$; \- \\
$N^n_{k+1} = \biguplus_{1 \leq i \leq p} S_i$;
\end{tabbing}
\scalebox{.8}{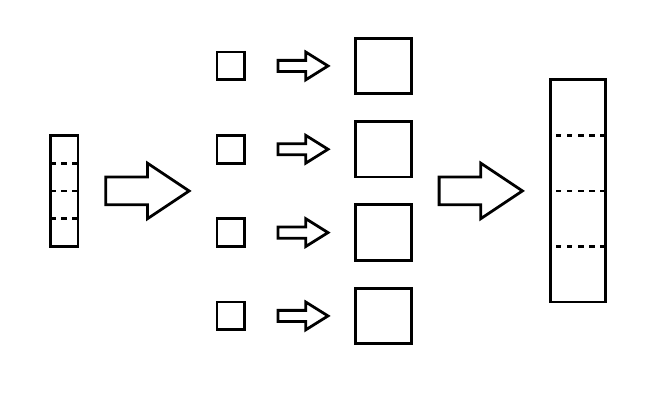}
\end{multicols}
\caption{Parallelization of \algorithm{Generate}.  The diagram on the right
  schematizes this process for $p=4$.}
\label{fig:par-generate}
\end{figure}

For \algorithm{Prune}, the parallelization is less trivial, as each
network from $N^n_k$ needs to be checked against all networks in the
current set of minimal (w.r.t.\ subsumption) networks.  In order to
make best use of the processing elements, we divide the parallel
execution into two phases. In the first phase, we split $N^n_k$ evenly
into $m\times p$ sets $S_1,\ldots,S_{m\times p}$, where for $m$ we
choose a multiplier for $p$ such that the individual sets have a
practically manageable size. Then we execute \algorithm{Prune} on these
sets in parallel.
In the second phase, for each set $S_i$ we still have to remove all networks that are subsumed
by networks in any other set $S_j$. To this end we define the algorithm \algorithm{Remove} (see Figure~\ref{fig:par-prune}),
which is a variant of \algorithm{Prune} where subsumption is only considered in one direction.

\begin{figure}
\begin{multicols}{2}
\paragraph{Algorithm \algorithm{Remove}.}\ \\
\begin{tabbing}
\textsf{input:} $S_i$ and $S_j$; \quad \textsf{output:} $S_i'$; \\
$S_i'=\emptyset$ \\
\textsf{for} \= $C\in S_i$ \textsf{do} \+ \\
  \textsf{for} \= $C'\in S_j$ \textsf{do} \+ \\
    \textsf{if} $(C' \preceq C)$ mark $C$ \- \\
\textsf{if} (not\_marked($C$)) \+ \\
    $S'_i =S'_i \cup \{C\}$; \\
\end{tabbing}

\paragraph{Algorithm \algorithm{Parallel-Prune}.}
\begin{tabbing}
\textsf{input:} $N^n_k$ and $m\times p$; \quad \textsf{output:} $R^n_k$; \\
split $N^n_k$ into sets $S_1,\ldots,S_{m\times p}$ \\
\textsf{for$||^p$} \= $i \in \{1,\ldots,m\times p\}$ \textsf{do} \+ \\
  $S_i = \algorithm{Prune}(S_i)$; \- \\
\textsf{for} $j \in \{1,\ldots,m\times p\}$ \textsf{do} \+ \\
  \textsf{for$||^p$} \= $i \in \{1,\ldots,m\times p\} \setminus \{j\}$ \textsf{do} \+ \\
    $S_i = \algorithm{Remove}(S_i,S_j)$; \- \- \\
$R^n_k = \biguplus_{1 \leq i \leq m\times p} S_i$;
\end{tabbing}
\end{multicols}

\vspace{2ex}
\centering
\scalebox{1}{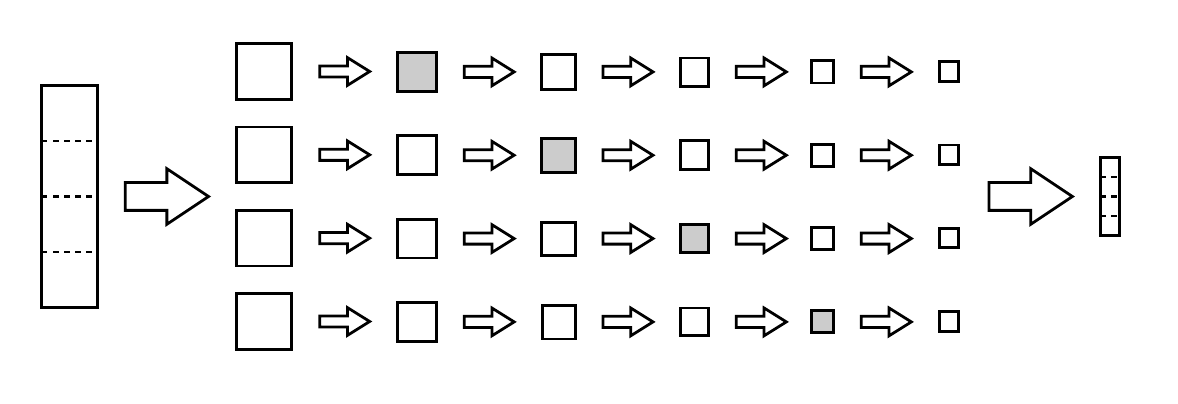}

\caption{Algorithms \algorithm{Remove} and \algorithm{Parallel-Prune} (top),
  and a graphical representation of the case $m=1$ and $p=4$ (bottom).
  At each stage of \algorithm{Prune}, the set $S_j$ is shaded.}
\label{fig:par-prune}
\end{figure}

After \algorithm{Remove} has finished, we replace set $S_i$ by the new (usually smaller) set $S_i'$.
Now, we observe that calling \algorithm{Remove} for sets $S_i$ and $S_j$ can
be performed in parallel to calling it for sets $S_k$ and $S_j$. Thus, in our parallelization
approach, we start by using the first set to remove networks from all other sets in parallel, then
we use the second set to remove networks from the first and all following sets, etc. After all sets
have been used in \algorithm{Remove}, the pruned
set $R^n_k$ can be obtained by merge-by-concatenation of all of
the final sets $S_i$.

The idea of the two phase version of \algorithm{Prune} is formalized
in the algorithm \algorithm{Parallel-Prune}, also detailed in
Figure~\ref{fig:par-prune}. This algorithm can be shown to have
isoefficiency $O(p^2\log^2p)$ using the techniques presented in
\cite{DBLP:journals/ieeecc/GramaGK93}, meaning that if we wanted to use
twice as many processors
maintaining efficiency, 
we would have to increase the problem
size by a factor a little greater than $4$.

In this way, $p$ processing elements can complete the first phase with
$m$ calls to \algorithm{Prune} per processing element.  The second
phase, with a total of $m\times p\times(m\times p-1)$ calls to
\algorithm{Remove}, requires approximately $m^2p^2$ calls
per processing element.  Although the comparisons in
\algorithm{Parallel-Prune} are not made in the same order as in the
original \algorithm{Prune}, experiments show that the total number of
comparisons made is roughly the same, while overhead grows with $m$.
Thus, in order to enhance overall performance, we can focus on
minimizing overhead, i.e., $m$ should be chosen to be minimal. In
other words, $m$ should be $1$ as long as the resulting sets
fit into memory for application of the \algorithm{Prune} algorithm. As
an additional measure to keep overhead low, minimum sizes of $1000$
and $5000$ comparator networks were imposed when splitting up the sets
in \algorithm{Parallel-Generate} and \algorithm{Parallel-Prune},
respectively.

The optimizations described in this section made it possible to
compute the sets $R^9_k$ for $1\leq k\leq 14$ in just over one
week\footnote{More precisely, in 7 days, 17 hours, and 58 minutes.}
using values $p = 288$ and $m=1$ in all runs of
\algorithm{Parallel-Prune}. 
The sizes of the sets $R^9_k$ are shown in Table~\ref{tab:r9sets}.
At this stage, $|R^9_{14}|=914{,}444$, and
we continued our efforts on two alternative paths. 
On one path, we continued to run the generate-and-prune approach to
compute $R^9_{k}$ for $15 \leq k \leq 25$. After five additional days
of computation, we obtained a singleton set for $R^9_{25}$
containing the minimal nine-channel sorting network depicted below.

\vspace{1ex}
\hfill\includegraphics{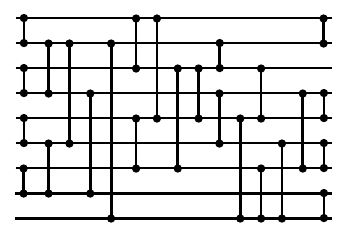}
\hspace*\fill\vspace{2ex}

\begin{table}
  \small
  \[\begin{array}{c|c|c|c|c|c|c|c|c|c|c|c|c|c}
    k & 1 & 2 & 3 & 4 & 5 & 6 & 7 & 8 & 9 & 10 & 11 & 12 & 13 \\ \hline
    |R^9_k| & 1 & 3 & 7 & 20 & 59 & 208 & 807 & 3{,}415 & 14{,}343 & 55{,}991 & 188{,}730 & 490{,}322 & 854{,}638
  \end{array}\]
  \[\begin{array}{c|c|c|c|c|c|c|c|c|c|c|c|c}
    k & 14 & 15 & 16 & 17 & 18 & 19 & 20 & 21 & 22 & 23 & 24 & 25 \\ \hline
    |R^9_k| & 914{,}444 & 607{,}164 & 274{,}212 & 94{,}085 & 25{,}786 & 5{,}699 & 1{,}107 & 250 & 73 & 27 & 8 & 1
  \end{array}\]

  \caption{Sizes of the sets $R^9_k$ for $1\leq k\leq 25$.}
  \label{tab:r9sets}
\end{table}

On the other path, we turned to consider the use of a SAT solver to
encode the search for an optimal-size sorting network on 9
channels. Given the set $R^9_{14}$, this required less than half a day of
computation on 288 threads (instead of 5 days), which is the topic of the next section.

\section{The SAT encoding approach}
\label{sec:step2}

In recent years, Boolean SAT-solving techniques have improved
dramatically, and SAT is currently applied to solve a wide variety of
hard and practical combinatorial problems, often outperforming
dedicated algorithms.
The general idea is to encode a hard problem instance, $\mu$, to a
Boolean formula, $\varphi_\mu$, such that the satisfying assignments
of $\varphi_\mu$ correspond to the solutions of $\mu$. Given such an
encoding, a SAT solver can be applied to solve $\mu$.
Recent attempts to
attack open instances of the optimal-depth sorting network problem,
such as those described in
\cite{DBLP:conf/mbmv/MorgensternS11,DBLP:conf/lata/BundalaZ14},
consider encodings to SAT. However, these encodings do not readily
apply to the optimal-size sorting network problem. In fact, we are not
aware of any previous attempts to encode the optimal-size sorting
network problem in SAT.

The encoding we propose in this paper is of size exponential in the
number of channels,~$n$. This is also the case for all previous SAT
encodings for the optimal-depth sorting network problem. Both of these
problems are naturally expressed in the form $\exists\forall\varphi$
(does there \emph{exist} a network that sorts \emph{all} of its
inputs?), and are easily shown to be in $\Sigma_2^P$.  We expect that,
similar to the problem of circuit minimization,
they are also $\Sigma_2^P$-hard,
although we have not succeeded to prove this. We do
not expect that there exists a polynomial-size encoding to SAT.

\subsection{Encoding the search for a sorting network}
We describe here a SAT encoding of the following decision problem,
which we term the $(n,k)$ sorting network problem: does there exist a
sorting network of size $k$ on $n$ inputs? We introduce this encoding
as a finite domain constraint model such that the encoding to conjunctive
normal form (CNF) of
each constraint in the model is straightforward. At the implementation
level, we apply the BEE compiler~\cite{BEE}, which performs this
encoding together with a range of ``compile-time'' optimizations.

We represent a size $k$ comparator network $\mathtt{Network}$ with $n$
channels as a sequence of the form $\mathtt{Network} =
\tuple{\mathtt{c(I_1,J_1)},\ldots,\mathtt{c(I_k,J_k)}}$ where the
$\mathtt{I_i}$ and $\mathtt{J_i}$ are finite domain integer variables
with domain $[1,n]$ and $\mathtt{I_i<J_i}$ for each $i$. The
conjunction of the following constraints encodes that
$\mathtt{Network}$ is a valid comparator network on $n$ channels.
\[ \mathtt{valid_n(Network)} = \bigwedge_{i=1}^{k} 
          \mathtt{new\_int(I_i,1,n)}\land\mathtt{new\_int(J_i,1,n)}\land
          \mathtt{int\_lt(I_i,J_i)}
\]
A constraint of the form $\mathtt{new\_int(I,1,n)}$ specifies that
$\mathtt{I}$ is the bit-level representation of an integer variable
with domain $[1,n]$. A constraint of the form $\mathtt{int\_lt(I,J)}$
specifies that the integer value represented by $\mathtt{I}$ is less
than that represented by $\mathtt{J}$. Below, we also consider the
constraint $\mathtt{eq(I,i)}$, which specifies that the integer value
represented by $\mathtt{I}$ is equal to the constant $i$.
The specific representation of integers is not important -- any of the
standard integer representations works. In our implementation, we adopt
a unary representation in the order encoding
(see e.g.~\cite{baker,DBLP:conf/cp/BailleuxB03}).

The conjunction of the following constraints encodes the impact of a
single comparator $\mathtt{c(I,J)}$ in terms of the vectors of Boolean
variables $\vec x=\tuple{x_1,\ldots,x_n}$ and $\vec
y=\tuple{y_1,\ldots,y_n}$, representing the values on the $n$ channels
before and after the comparator.
The first conjunction, $\varphi_{I,J}(\vec x,\vec y)$, specifies that
when integer variables $(I,J)$ take the values $(i,j)$, then $y_i =
x_i\land x_j$ and $y_j = x_i\lor x_j$, i.e., the minimum goes to $y_i$ and the
maximum to $y_j$.
The second conjunction, $\psi_{I,J}(\vec x,\vec y)$, specifies that $x_i=y_j$ for
all channels $i$ different from the values $I$ and $J$.
\begin{align*}
\varphi_{I,J}(\vec x,\vec y) &=
       \bigwedge_{1\leq i<j\leq n}
       \left(
           \mathtt{int\_eq(I,i)}\land \mathtt{int\_eq(J,j)} \rightarrow%
           (y_i\leftrightarrow x_i\land x_j) \land 
                     (y_j\leftrightarrow x_i\lor x_j)
       \right)
\\
\psi_{I,J}(\vec x,\vec y) &=
     \bigwedge_{1\leq i\leq n} 
     \left(
                   \neg\mathtt{int\_eq(I,i)} \land
                   \neg\mathtt{int\_eq(J,i)} \rightarrow %
                   x_i\leftrightarrow y_i
     \right)
\end{align*}
The following encodes that $\mathtt{Network} =
\tuple{\mathtt{c(I_1,J_1)},\ldots,\mathtt{c(I_k,J_k)}}$ sorts $\vec
b\in \BB^n$. 
Let $\vec x_0 = \vec b$, $\vec x_k$ be equal to the vector obtained by
sorting $\vec b$, and let $\vec x_1,\ldots,\vec x_{k-1}$ be length $n$
vectors of Boolean variables. Then,
\[ \mathtt{sorts(Network,\mbox{$\vec{b})$}} = \bigwedge_{i=1}^{k} 
       \varphi_{I_i,J_i}(\vec x_{i-1},\vec x_{i}) \land 
       \psi_{I_i,J_i}(\vec x_{i-1},\vec x_{i})
\]
A sorting network with $k$ comparators on $n$ channels must sort all of
its inputs. Hence, a sorting network with $k$ comparators on $n$ channels
exists if and only if the following formula is satisfiable.

\begin{equation}
\label{satEncoding1}
  \mathtt{sorter_n}(\mathtt{Network}) =  \mathtt{valid_n(Network)} \land
    \bigwedge_{{\vec b}\in\BB^n}\mathtt{sorts(Network,\mbox{$\vec b$})}
\end{equation}

Our implementation of the above encoding introduces several additional
optimizations. We list these here briefly, for $\mathtt{Network} =
\tuple{\mathtt{c(I_1,J_1)},\ldots,\mathtt{c(I_k,J_k)}}$.
\begin{itemize}

\item\emph{No redundant neighbors.} For each $1\leq i<k$, we add the
  constraint: $I_i\neq I_{i+1} \lor J_i\neq J_{i+1}$.
\item\emph{Independent comparators in ascending order.} For each $1\leq i<k$, we add the
  constraint: $I_i\neq I_{i+1} \land I_i\neq J_{i+1} \land J_i\neq I_{i+1} \land J_i\neq J_{i+1} \rightarrow
  I_i<I_{i+1}$.
\item\emph{All adjacent comparators.} Following Exercise 5.3.4.35
  of~\cite{Knuth73}, we add the constraint that states that
  all comparators of the form $(i,i+1)$ must be present in every
  standard sorting network.
\item\emph{Only unsorted inputs.} Let $\BB^n_{un}$ denote the
  subset of $\BB^n$ consisting of unsorted sequences. Then it is
  possible to refine the conjunction in Equation~(\ref{satEncoding1})
  replacing $\BB^n$ with the smaller $\BB^n_{un}$. Moreover, observe
  that $|\BB^n_{un}| = 2^n-n-1$, and as noted by Chung and Ravikumar in
  \cite{Chung1990}, this is the size of the smallest test set possible
  in order to determine that $\mathtt{Network}$ is a sorting network.
\end{itemize}

Table~\ref{tab:sat} shows the results obtained with our implementation
of the SAT encoding described above. The left part of the table
concerns the search for sorting networks of optimal size; and the
right part, the ``proof'' that smaller networks do not exist. The
columns labeled ``BEE'' detail the compilation times (in seconds) to
generate the CNF and to perform optimizations prior to SAT
solving. The columns labeled ``SAT'' detail the SAT-solving times (in
seconds) for the satisfiable instances, on the left, and for the
unsatisfiable instances, on the right. The $\infty$ symbol indicates a
time-out: these instances did not terminate even after one week of
computation. We observe that the sizes of these SAT instances, even
those that we cannot solve, are not excessive: all instances contain
less than one million clauses, and less than one quarter of a million
variables.

\begin{table}[hb]
  \centering
\begin{tabular}{|l||r|r|r|r|r||r|r|r|r|r||}
\cline{2-11} \multicolumn{1}{c||}{}
    &\multicolumn{5}{c||}{optimal sorting networks (sat)}      
    &\multicolumn{5}{c||}{smaller networks (unsat)}\\
\hline
$n$ &$k$&
\multicolumn1{c|}{BEE} &
\multicolumn1{c|}{\#clauses} &
\multicolumn1{c|}{\#vars} &
\multicolumn1{c||}{SAT} &
$k$&
\multicolumn1{c|}{BEE} &
\multicolumn1{c|}{\#clauses} &
\multicolumn1{c|}{\#vars} &
\multicolumn1{c||}{SAT} \\
\hline
  4 &5 &  0.18 &  1916&   486& 0.01     &4 & 0.15 &  1480 & 356 & 0.01     \\
  5 &9 &  1.03 & 10159&  2550& 0.03     &8 & 0.90 &  8963 &2221 & 1.27     \\
  6 &12&  4.55 & 35035&  8433& 2.45     &11& 3.99 & 32007 &7657& 242.02    \\
  7 &16& 21.68 &114579& 26803&16.70     &15& 19.04 &107227& 25000&\cinfty\\
  8 &19& 82.93 &321445& 73331&\cinfty   &18& 73.34 &304145& 69221&\cinfty\\
  9 &25&452.55 &977559&219950&\cinfty   &24& 406.67 &937773&210715&\cinfty\\
\hline
\end{tabular}
  \caption{SAT solving for $n$-channel sorting networks with $k$
    comparators: BEE compile times and SAT solving times are in seconds. }
  \label{tab:sat}
\end{table}

The optimal sorting networks for sizes $5$ to $7$ found by this algorithm
are represented below.

\vspace{1ex}\noindent\hfill
\raisebox{-\height/2}{\includegraphics{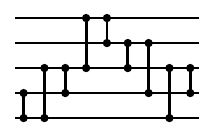}}
\hfill
\raisebox{-\height/2}{\includegraphics{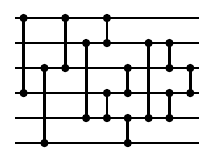}}
\hfill
\raisebox{-\height/2}{\includegraphics{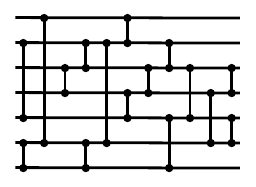}}
\hspace*\fill\vspace{2ex}

\subsection{Searching from a complete set of comparator networks}

Since the methodology presented above does not scale beyond 
$n=6$, we will
now show how to capitalize on the results from Section~\ref{sec:implement}.
Therefore, we focus on the following variant of the previous problem,
which we term the
$(n,k,S)$ sorting network problem: given a (complete) set of comparator
networks, $S$ on $n$ channels, is there a network $C\in S$ that can be
extended to a sorting network of size $k$?

To solve this problem, we consider each element $C\in S$ separately. We
encode the corresponding $(n,k)$ sorting network problem in terms of
$\mathtt{Network} =
\tuple{\mathtt{c(I_1,J_1)},\ldots,\mathtt{c(I_k,J_k)}}$, and we fix the
values of the comparator positions in the prefix of $\mathtt{Network}$
to match the positions of those in $C$.  Even this small difference
turns out to provide one key ingredient to solve the optimal-size
sorting network problem; the other key ingredient is to make
sure that the size of the set $S$ is as small as possible.

With the SAT encoding of Equation~(\ref{satEncoding1}), we are not able to show
that there is no sorting network of size 15 on 7 channels (even given
a week of computation time). Recall Lemma~\ref{lem:complete}, and
consider $n=7$. The set $R^7_3$ consists of $7$ comparator networks
and is complete. So, there exists an optimal-size sorting network on
$7$ channels if and only if there exists one of the form $C;C'$ for
some $C\in R^7_3$.
Solving the $(7,15,R^7_3)$ sorting network problem reveals that there
is no sorting network on 7 channels with 15 comparators. The
computational cost of this proof sums up to approximately 10 minutes
of parallel computation (on 7 cores), or less than 1 hour in total of
sequential computation.

Solving the SAT and UNSAT cases for 8 channels is more involved. Here
we consider $R^8_5$, which is a complete set of comparator networks
with 5 comparators each and consists of 57 elements.
For the UNSAT case, computation requires just under 1.36
hours on 57 cores (the time to complete the slowest instance), or a
total of 33.83 hours on a single core.
For the SAT case, computation requires 0.35 of an hour (on 57 cores),
which is the time until the first satisfiable instance terminates.

There is one further optimization, adopted from
\cite{DBLP:conf/lata/BundalaZ14}, that we consider when encoding the
search for a sorting network that extends a given comparator
network. Consider again Equation~(\ref{satEncoding1}). A sorting
network must sort all of its (unsorted) inputs and hence the
conjunction of all $\vec b\in\BB^n$ (or the smaller set
$\BB^n_{un}$). However, if we consider any specific subset of
$B\subseteq\BB^n$ and show that there is no comparator network that
sorts the elements of $B$, then surely there is also no comparator
network that sorts the elements of $\BB^n$.
In particular, we consider the set $\BB^n_s$, which we call the set of
\emph{windows of size $s$}, of all unsorted length $n$ binary
sequences of the form $0^{\ell_1}.w.1^{\ell_2}$ such that
$\ell_1+\ell_2 = s$.  If the encoding of Equation~(\ref{satEncoding1})
is unsatisfiable when replacing $\BB^n$ with $\BB^n_s$, then it is
unsatisfiable also in its original form. Solving the UNSAT case for
$8$ channels and 18 comparators using this optimization reduces the total
solving time from 33.83 hours to 27.52 hours. From the 57 instances
that need be shown unsatisfiable, 50 are found so with $s=3$; a
further 4 with $s=2$; and the remaining 3 with $s=1$.

To solve the optimal-size sorting network problem for $n=9$ channels,
we consider the 
($9,24,R^9_{14}$) sorting network problem, where $R^9_{14}$ is the
complete set of $914{,}444$ comparator networks obtained using the
technique described in Section~\ref{sec:implement}. We show that each
of the corresponding propositional formulae is unsatisfiable, implying
that there is no extension of an element of  $R^9_{14}$ to a 24
comparator sorting network. 
The average solving time (per instance) is 4.09 seconds for
compilation and 7.83 seconds for the SAT solver.  The total solving
time for all instances (compilation and SAT solving) is 3028 hours.
There is an additional overhead of 333 hours for using the windows
optimization (the cost of resolving with a smaller window when an
instance is satisfiable). Running 288 
threads on 144 cores requires
just under 12 hours of computation.
From the 914{,}444 instances, 675{,}736 (74\%) were found
unsatisfiable using a window of size 4, 233{,}400 (25\%) were found
unsatisfiable using a window of size 3, 4{,}979 (less than 1\%) were
found unsatisfiable using a window of size 2, and the remaining 329
(less than 1\%) were found unsatisfiable using a window of size 1.

\section{Proving optimality of 25 comparators for 9 inputs}
\label{sec:discussion}

In Sections~\ref{sec:implement} and~\ref{sec:step2} we provide two
alternative proofs that $S(9)=25$, both of which rely on first computing the
set $R^9_{14}$ (1 week of computation).
For the first alternative, using the techniques of
Section~\ref{sec:implement} we apply the generate-and-prune approach
continuing from $R^9_{14}$ until termination with $|R^9_{25}|=1$ (an
additional 5 days of computation). 
For the second alternative, using the techniques of
Section~\ref{sec:step2}, we apply a SAT solver to solve the
$(9,24,R^9_{14})$ sorting network problem, showing that no
element of $R^9_{14}$ can be extended to a $24$-comparator sorting network
(less than half a day of computation).

For both alternatives, the implementation relies on a Prolog program to
compute the sets $R^9_{k}$.  The second alternative involves also a
Prolog implementation of the SAT encoding, the BEE constraint
compiler, and the state-of-the-art SAT solver CryptoMiniSAT~\cite{Crypto}.
We use SWI-Prolog 6.6.1.

While it is reassuring to have two alternative proofs, they both share
the computation of $R^9_{14}$. Although we have proved all of the
mathematical claims underlying the design of the proof algorithm
and have carefully checked the correctness of the Prolog implementation,
there is always the potential for errors in computer programs.  The
objective of this section is to provide further confidence in the
correctness of our results.

One of the key aspects of computer-assisted proofs is guaranteeing
their validity.  Barendregt and Wiedijk~\cite{Barendregt2005}
introduced the \emph{de Bruijn criterion}: every computer-assisted
proof should be verifiable by an independent small program (a
``verifier'').  In this section, we summarize how our approach meets
this criterion.

Verifiers for SAT encodings are, in our case, more complex, as the
instances we need to verify are all unsatisfiable. While satisfiable
instances have concrete assignments as their witnesses, for
unsatisfiable instances we would have to verify $914{,}444$ (minimal)
unsatisfiable cores. Hence, we focus our validity argument on the
generate-and-prune approach, which involves two critical points. We must
guarantee that: (1) when extending a network with $k$ comparators to
one with $k+1$ comparators, all extensions are considered, and (2)
when eliminating a network, this is sound.  

In order to verify our result independently from the Prolog
implementation, the code is augmented to produce a log file during
execution. 
We then verify that the information in this file provides a sound and
complete basis for the reconstruction of our proof that there is no
9-channel sorting network consisting of 24 comparators. To this end,
an independent Java implementation of the generate-and-prune algorithm
is provided, with one main important difference to the Prolog
implementation: it performs no search, and is aware only of the
information available in the log file. 

The log file contains lines of the form
``$\mathtt{killed(C_1,C_2,\pi)}$'', specifying that network $C_1$ is
pruned because it is subsumed by a network $C_2$ with permutation
$\pi$ (namely, that $C_2 \leq_\pi C_1$).  Such lines are introduced
both when extending a network with a redundant comparator (here the
permutation is an identity), as well as when pruning.

The verifier reconstructs the computation of all of the sets $R^9_k$,
starting from $R^9_0$ which consists of the empty comparator
network. When extending $R^9_k$ to $R^9_{k+1}$ it first performs a
naive extension to $N^9_{k+1}$, adding all comparators in all possible
positions, and then computes $R^9_{k+1}$ using the log file
only. Namely, for each row of the form $\mathtt{killed(C_1,C_2,\pi)}$,
we first verify that indeed $\pi(\outputs(C_2))\subseteq\outputs(C_1)$,
and then remove $C_1$ from $N^9_{k+1}$. 
By soundness, we mean that whenever a network is eliminated, we have
verified that the logged permutation $\pi$ is indeed a witness to its
redundancy. 
By completeness, we mean that after pruning we have a complete set of
comparator networks.  In order to ensure completeness, we additionally
verify that the logged subsumption information is acyclic. Otherwise,
it would be possible, for example, that there were two networks, $C_1$ and
$C_2$ such that both $C_1 \leq_{\pi_1} C_2$ and $C_2 \leq_{\pi_2} C_1$,
and that both were eliminated.

Using this tool, we verified the computer-assisted proof of $n=7$ in
$4$ seconds, the one for $n=8$ in $2$ minutes, and the one for $n=9$
in just over $6$ hours of computational time. The logs and the Java
verifier are available from: \url{http://imada.sdu.dk/~petersk/sn/}

\section{Conclusions}
\label{sec:concl}

We have shown that $S(9) = 25$, i.e., the minimal number of
comparators needed to sort nine inputs is 25. This closes the smallest
open instance of the optimal-size problem for sorting networks,
which was open since~1964. As a
corollary, given the result from~\cite{Knuth73} that states that
$S(10)\leq 29$, and applying the inequality $S(n+1) \geq S(n) +
\lceil\log_2n\rceil$ from \cite{voorhis72}, we now also know that
$S(10) = 29$.

\section*{Acknowledgement}
We thank Carsten Fuhs and Donald E. Knuth for their constructive comments on draft versions of this paper.

\bibliographystyle{abbrv}
\bibliography{size}

\begin{thebibliography}{10}

\bibitem{Barendregt2005}
H.~Barendregt and F.~Wiedijk.
\newblock The challenge of computer mathematics.
\newblock {\em Transactions A of the Royal Society}, 363(1835):2351--2375,
  2005.

\bibitem{DBLP:conf/afips/Batcher68}
K.~E. Batcher.
\newblock Sorting networks and their applications.
\newblock In {\em AFIPS Spring Joint Computing Conference}, volume~32 of {\em
  AFIPS Conference Proceedings}, pages 307--314. Thomson Book Company,
  Washington D.C., 1968.

\bibitem{DBLP:conf/lata/BundalaZ14}
D.~Bundala and J.~Z{\'a}vodn{\'y}.
\newblock Optimal sorting networks.
\newblock In {\em LATA 2014}, LNCS 8370, pages 236--247. Springer, 2014.

\bibitem{Chung1990}
M.~J. Chung and B.~Ravikumar.
\newblock Bounds on the size of test sets for sorting and related networks.
\newblock {\em Discrete Mathematics}, 81(1):1 -- 9, 1990.

\bibitem{ourarxivstuff}
M.~Codish, L.~Cruz-Filipe, and P.~Schneider-Kamp.
\newblock The quest for optimal sorting networks: Efficient generation of
  two-layer prefixes.
\newblock {\em CoRR}, abs/1404.0948, 2014.

\bibitem{Knuth66}
R.~W. Floyd and D.~E. Knuth.
\newblock The {B}ose--{N}elson sorting problem.
\newblock In J.~Srivastava, editor, {\em A survey of combinatorial theory},
  pages 163--172. North-Holland, 1973.

\bibitem{DBLP:journals/ieeecc/GramaGK93}
A.~Grama, A.~Gupta, and V.~Kumar.
\newblock Isoefficiency: measuring the scalability of parallel algorithms and
  architectures.
\newblock {\em IEEE P{\&}DT}, 1(3):12--21, 1993.

\bibitem{Knuth73}
D.~E. Knuth.
\newblock {\em The Art of Computer Programming, Volume III: Sorting and
  Searching}.
\newblock Addison-Wesley, 1973.

\bibitem{BEE}
A.~Metodi, M.~Codish, and P.~J. Stuckey.
\newblock Boolean equi-propagation for concise and efficient sat encodings of
  combinatorial problems.
\newblock {\em J. Artif. Intell. Res. (JAIR)}, 46:303--341, 2013.

\bibitem{DBLP:conf/mbmv/MorgensternS11}
A.~Morgenstern and K.~Schneider.
\newblock Synthesis of parallel sorting networks using {SAT} solvers.
\newblock In {\em MBMV 2011}, pages 71--80. OFFIS-Institut f{\"u}r Informatik,
  2011.

\bibitem{DBLP:journals/mst/Parberry91}
I.~Parberry.
\newblock A computer-assisted optimal depth lower bound for nine-input sorting
  networks.
\newblock {\em Mathematical Systems Theory}, 24(2):101--116, 1991.

\bibitem{DBLP:conf/parle/Parberry91}
I.~Parberry.
\newblock On the computational complexity of optimal sorting network
  verification.
\newblock In E.~H.~L. Aarts, J.~van Leeuwen, and M.~Rem, editors, {\em PARLE
  (1)}, volume 505 of {\em Lecture Notes in Computer Science}, pages 252--269.
  Springer, 1991.

\bibitem{voorhis72}
D.~Voorhis.
\newblock Toward a lower bound for sorting networks.
\newblock In R.~Miller, J.~Thatcher, and J.~Bohlinger, editors, {\em Complexity
  of Computer Computations}, The IBM Research Symposia Series, pages 119--129.
  Springer US, 1972.

\end{thebibliography}

\end{document}